\documentclass[a4paper,UKenglish,cleveref, autoref, thm-restate]{lipics-v2021}

\usepackage{graphicx}
\usepackage[all]{xy}

\usepackage{amsmath,amssymb,amsfonts,stmaryrd}
\usepackage{quiver}

\usepackage[utf8]{inputenc}

\usepackage[margin=8pt,font=small,labelfont=bf, labelsep=endash]{caption}

\usepackage{tikzit}

\tikzstyle{observation}=[text=blue]
\tikzstyle{morphism}=[fill=white, draw=black, shape=rectangle]
\tikzstyle{medium box}=[fill=white, draw=black, shape=rectangle, minimum width=0.8cm, minimum height=0.9cm]
\tikzstyle{large morphism}=[fill=white, draw=black, shape=rectangle, minimum width=1.7cm, minimum height=1cm]
\tikzstyle{bn}=[fill=black, draw=black, shape=circle, inner sep=1.5pt]
\tikzstyle{state}=[fill=white, draw=black, regular polygon, regular polygon sides=3, minimum width=0.8cm, shape border rotate=180, inner sep=0pt]
\tikzstyle{effect}=[fill=white, draw=black, regular polygon, regular polygon sides=3, minimum width=0.8cm, shape border rotate=0, inner sep=0pt]
\tikzstyle{medium state}=[fill=white, draw=black, regular polygon, regular polygon sides=3, minimum width=1.3cm, inner sep=0pt, shape border rotate=180]
\tikzstyle{medium effect}=[fill=white, draw=black, regular polygon, regular polygon sides=3, minimum width=1.3cm, inner sep=0pt, shape border rotate=0]
\tikzstyle{large state}=[fill=white, draw=black, regular polygon, regular polygon sides=3, minimum width=2.2cm, shape border rotate=180, inner sep=0pt]
\tikzstyle{wn}=[fill=white, draw=black, shape=circle, inner sep=1.5pt]
\tikzstyle{treenode}=[fill=white, draw=none, shape=circle]

\tikzstyle{arrow}=[->]
\tikzstyle{dashed box}=[-, dashed]
\tikzstyle{condition}=[draw=blue, dashed]

\pgfdeclarelayer{edgelayer}
\pgfdeclarelayer{nodelayer}
\pgfsetlayers{edgelayer,nodelayer,main}
\tikzstyle{none}=[]

\tikzstyle{observation}=[text=blue]
\tikzstyle{morphism}=[fill=white, draw=black, shape=rectangle]
\tikzstyle{medium box}=[fill=white, draw=black, shape=rectangle, minimum width=0.8cm, minimum height=0.9cm]
\tikzstyle{large morphism}=[fill=white, draw=black, shape=rectangle, minimum width=1.7cm, minimum height=1cm]
\tikzstyle{bn}=[fill=black, draw=black, shape=circle, inner sep=1.5pt]
\tikzstyle{state}=[fill=white, draw=black, regular polygon, regular polygon sides=3, minimum width=0.8cm, shape border rotate=180, inner sep=0pt]
\tikzstyle{effect}=[fill=white, draw=black, regular polygon, regular polygon sides=3, minimum width=0.8cm, shape border rotate=0, inner sep=0pt]
\tikzstyle{medium state}=[fill=white, draw=black, regular polygon, regular polygon sides=3, minimum width=1.3cm, inner sep=0pt, shape border rotate=180]
\tikzstyle{medium effect}=[fill=white, draw=black, regular polygon, regular polygon sides=3, minimum width=1.3cm, inner sep=0pt, shape border rotate=0]
\tikzstyle{large state}=[fill=white, draw=black, regular polygon, regular polygon sides=3, minimum width=2.2cm, shape border rotate=180, inner sep=0pt]
\tikzstyle{wn}=[fill=white, draw=black, shape=circle, inner sep=1.5pt]
\tikzstyle{treenode}=[fill=white, draw=none, shape=circle]

\tikzstyle{arrow}=[->]
\tikzstyle{dashed box}=[-, dashed]
\tikzstyle{condition}=[draw=blue, dashed]

\tikzset{baseline=(current  bounding  box.center)}

\usepackage{color}
\definecolor{deepblue}{rgb}{0,0,0.5}
\definecolor{deepred}{rgb}{0.6,0,0}
\definecolor{deepgreen}{rgb}{0,0.5,0}
\definecolor{darkgray}{rgb}{0.5,0.5,0.5}

\usepackage{listings}

\DeclareFixedFont{\ttb}{T1}{txtt}{bx}{n}{9} 
\DeclareFixedFont{\ttm}{T1}{txtt}{m}{n}{9}  

\lstdefinestyle{python_ppl}{
	language=Python,
	basicstyle=\ttm,
	otherkeywords={let, to},          
	keywordstyle=\ttb\color{deepblue},
	emph={Gauss,N,condition, =:=, observe, sample, score, normal, gp_sample},     
	emphstyle=\ttb\color{deepred},    
	stringstyle=\color{deepgreen}  
}

\lstdefinelanguage{CustomML}{
	keywords={match, with, rec, true, false, fun, return, let, in, if, then, else, type, val, module, sig, end, ref, struct},
	keywordstyle=\color{deepblue}\bfseries,
	ndkeywords={ref},
	ndkeywordstyle=\color{darkgray}\bfseries,
	identifierstyle=\color{black},
	sensitive=false,
	comment=[l]{//},
	morecomment=[s]{/*}{*/},
	commentstyle=\color{darkgray}\ttfamily,
	stringstyle=\color{red}\ttfamily,
	morestring=[b]',
	morestring=[b]"
}

\lstdefinestyle{ml_ppl}{
	language=CustomML,
	basicstyle=\ttm,
	otherkeywords={},          
	keywordstyle=\ttb\color{deepblue},
	emph={Gauss,N,condition, =:=, observe, sample, score, normal, gp_sample},     
	emphstyle=\ttb\color{deepred},    
	stringstyle=\color{deepgreen}
}

\lstdefinelanguage{JavaScript}{
	keywords={typeof, new, true, false, catch, function, return, null, catch, switch, var, if, in, while, do, else, case, break},
	keywordstyle=\color{blue}\bfseries,
	ndkeywords={class, export, boolean, throw, implements, import, this},
	ndkeywordstyle=\color{darkgray}\bfseries,
	identifierstyle=\color{black},
	sensitive=false,
	comment=[l]{//},
	morecomment=[s]{/*}{*/},
	commentstyle=\color{darkgray}\ttfamily,
	stringstyle=\color{red}\ttfamily,
	morestring=[b]',
	morestring=[b]"
}

\lstdefinestyle{webppl}{
	language=JavaScript,
	basicstyle=\ttm,
	otherkeywords={let},            
	keywordstyle=\ttb\color{deepblue},
	emph={flip, condition, factor, sample, normal, score, observe},
	emphstyle=\ttb\color{deepred},    
	stringstyle=\color{deepgreen}
}

\lstdefinestyle{prolog}{
	language=Prolog,
	basicstyle=\ttm,         
	keywordstyle=\ttb\color{deepblue},
	emphstyle=\ttb\color{deepred},    
	stringstyle=\color{deepgreen}, 
}

\lstdefinestyle{funprolog}{
	language=Prolog,
	basicstyle=\ttm,
	otherkeywords={let, in, return},            
	keywordstyle=\ttb\color{deepblue},
	emphstyle=\ttb\color{deepred},    
	stringstyle=\color{deepgreen}, 
}

\newcommand*{\mlstinline}[1]{\text{\lstinline|#1|}}
\newcommand*{\code}[1]{\lstinline|#1|}

\newcommand{\cat}{\mathbb}
\newcommand{\catname}{\mathsf}
\newcommand{\C}{\cat C}

\newcommand{\id}{\mathsf{id}}
\newcommand{\supp}{\mathsf{supp}}

\newcommand{\R}{\mathbb R}

\newcommand{\N}{\mathcal N}

\newcommand{\eq}{\mathrel{\scalebox{0.4}[1]{${=}$}{:}{\scalebox{0.4}[1]{${=}$}}}}

\newcommand{\defeq}{\stackrel{\text{def}}=}

\newcommand{\cond}{\catname{Cond}}
\newcommand{\gauss}{\catname{Gauss}}
\newcommand{\gaussex}{\catname{GaussEx}}
\newcommand{\gaussrel}{\catname{GaussRel}}
\newcommand{\linrel}{\catname{LinRel}}

\newcommand{\vect}{\mathsf{Vec}}
\newcommand{\cmon}{\mathsf{CMon}}
\newcommand{\aff}{\catname{Aff}}

\newcommand{\col}{\mathsf{col}}

\newcommand{\lin}[1]{\mathsf{Lin}_{#1}}

\newcommand{\hide}[1]{{}}

\title{A Category for unifying Gaussian Probability and Nondeterminism} 

\titlerunning{Combining Gaussian Probability with Nondeterminism} 

\author{Dario {Stein}}{iHub, Radboud University Nijmegen, The Netherlands}{dario.stein@ru.nl}{}{}
\author{Richard {Samuelson}}{Humming Inc., United States of America}{richard@heyhumming.com}{}{}

\authorrunning{D. Stein} 

\Copyright{Dario Stein, Richard Samuelson} 



\category{} 

\relatedversion{} 

\acknowledgements{It has been useful to discuss this work with many people. Particular thanks go to Tobias Fritz, Bart Jacobs, Dusko Pavlovic, Sam Staton and Alexander Terenin.}

\relatedversiondetails[cite={stein2022decorated}]{Previous Version}{https://arxiv.org/abs/2204.14024v1}

\nolinenumbers 

\EventEditors{Paolo Baldan and Valeria de Paiva}
\EventNoEds{2}
\EventLongTitle{10th Conference on Algebra and Coalgebra in Computer Science (CALCO 2023)}
\EventShortTitle{CALCO 2023}
\EventAcronym{CALCO}
\EventYear{2023}
\EventDate{June 19--21, 2023}
\EventLocation{Indiana University Bloomington, IN, USA}
\EventLogo{}
\SeriesVolume{270}
\ArticleNo{13}

\begin{document}

\maketitle

\begin{abstract} 
We introduce categories of extended Gaussian maps and Gaussian relations which unify Gaussian probability distributions with relational nondeterminism in the form of linear relations. Both have crucial and well-understood applications in statistics, engineering, and control theory, but combining them in a single formalism is challenging. It enables us to rigorously describe a variety of phenomena like noisy physical laws, Willems' theory of open systems and uninformative priors in Bayesian statistics. The core idea is to formally admit vector subspaces $D \subseteq X$ as generalized uniform probability distribution. Our formalism represents a first bridge between the literature on categorical systems theory (signal-flow diagrams, linear relations, hypergraph categories) and notions of probability theory. 
\end{abstract}

\section{Introduction}

Modelling the behavior of systems under uncertainty is of crucial importance in engineering and computer science. We can distinguish two different kinds of uncertainty:
\begin{itemize}
\item Probabilistic uncertainty means we may not know the exact value of some quantity, like a measurement error, but we do know the statistical distribution of such errors. A typical such distribution is the \emph{normal (Gaussian) distribution} $\mathcal N(\mu,\sigma^2)$ of mean $\mu$ and variance~$\sigma^2$.
\item Nondeterministic uncertainty models complete ignorance of a quantity. We know which values the quantity may feasibly assume but have no statistical information beyond that. Nondeterministic uncertainty can be modelled using subsets $R \subseteq X$ which identify the feasibles values. In practice, such subsets are often characterized by equational constraints such as natural laws.
\end{itemize}

\noindent Systems may be subject to both probabilistic and nondeterministic constraints, but describing such systems mathematically is more challenging. A classical treatment is Willems' theory of \emph{open stochastic systems} \cite{willems:oss,willems:constrained}, where `openness' in his terminology refers to nondeterminism or lack of information. We recall a simple example:


\begin{example}[Noisy resistor]\label{ex:resistor}
For a resistor of resistance $R$, Ohm's law constrains pairs $(V,I)$ of voltage and current to lie in the subspace $D = \{ (V,I) : V = RI \}$. This is a relational constraint -- values must lie in $D$, but we have no further statistical information about which values the system takes. In a realistic system, thermal noise is always present; such a noisy system is better modelled by the equation
\begin{equation}
V = RI + \epsilon
\end{equation}
where $\epsilon \sim \N(0,\sigma^2)$ is a Gaussian random variable with some small variance $\sigma^2$. Willems notices that the variables $V,I$ are not random variables in the usual sense; we have not associated any distribution to them. On the other hand, the quantity $V-RI$ is a honest random variable. Furthermore, if we supply a fixed voltage $V_0$, we can solve for $I$ and
\begin{equation}
I = R^{-1}(V_0-\epsilon) \label{eq:interconnection}
\end{equation}becomes a classical (Gaussian) random variable. Willems calls this `interconnection' of systems.
\end{example}
Willems models the `openness' of the stochastic systems by endowing the outcome space $\R^2$ with an unusually coarse $\sigma$-algebra $\mathcal E$ to formalize the lack of information. Measurable sets are restricted to the form $\{(V,I) : V-RI \in A \}$ for $A \subseteq \R$ Borel. The Gaussian probability measure is then only defined on $\mathcal E$, which essentially makes it a measure on the quotient space $\R^2/D$. We purely formally define an \emph{extended Gaussian distributions} on a space $X$ as a pair $(D,\psi)$ of a subspace $D$ and a Gaussian distribution on $X/D$. In particular, we can think of any subspace $D$ as an extended Gaussian distribution $(D,0)$. Operationally, sampling a point $x \sim D$ means picking it nondeterministically from $D$. Every extended Gaussian distribution can be seen as a formal sum $\psi + D$ of a Gaussian contribution $\psi$ and a nondeterministic contribution $D$.

In our approach, the noisy resistor is described by a single extended Gaussian distribution where $D$ is the subspace for for Ohm's law, and $\psi$ is Gaussian noise in a direction orthogonal to $D$. The marginals $V,I$ are themselves extended Gaussian distributions: we find that $V \sim \R$ and $I \sim \R$, that is they are picked nondeterministically from the real line, so in this sense we have no information about them. We also find that $V-RI \sim \N(0,\sigma^2)$ follows a classical Gaussian distribution without any nondeterministic contribution. The interconnection \eqref{eq:interconnection} is obtained as an instance of probabilistic conditioning $V=V_0$. We compare our approach to the one of Willems in \Cref{sec:open_system}. \\

\noindent We now describe a completely different situation where it makes sense to admit subspaces as idealized probability distributions, namely uninformative priors in Bayesian inference:

\begin{figure}
	\includegraphics[width=0.4\textwidth]{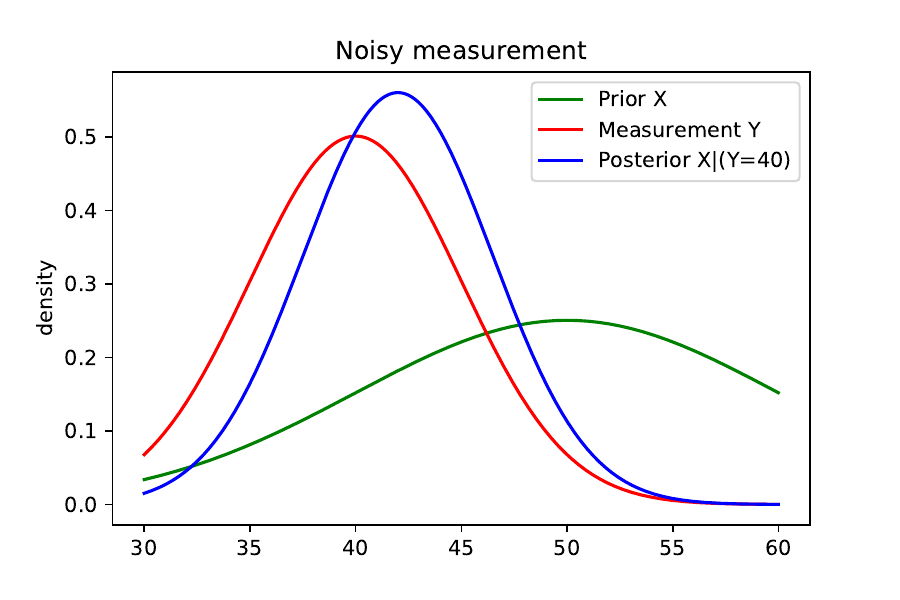}
	\caption{Gaussian prior and posterior in a noisy measurement example \label{fig:noisymmt}}
\end{figure}

\begin{example}[Uninformative Priors]\label{ex:inference}
Our prior experience tell us that we expect the mass $X$ of some object to be normally distributed with mean $50$ and variance $100$. We use a noisy scale to obtain a measurement of $Y=40$. If the scale error has variance of $25$, we can compute our \emph{posterior} belief over $X$, which turns out to be \footnote{see appendix \ref{sec:mmt} for the calculation} $X|(Y = 40) \sim \N(42,20)$. Here, the influence of the prior has corrected the predicted value to lie slightly above the measured value, and have smaller overall variance (see \Cref{fig:noisymmt}).

If we had no prior information at all about $X$, the posterior should simply reflect the measurement uncertainty $\N(40,25)$. We can model this by putting a larger and larger variance on $X$. However, the limit of distributions $\N(50,\sigma^2)$ for $\sigma^2 \to \infty$ does not exist in any measure-theoretic sense, because it would approach the zero measure on every measurable set. There exists no uniform probability distribution over the real line. In practice, one can sometimes pretend (using the method of \emph{improper priors}, e.g. \cite{gelman,hedegaard:gaussian_random_fields}) that $X$ is sampled from the Lebesgue measure $\lambda$ (with constant density $1$). This measure fails to be normalized, however the resulting density calculations may yield the correct probability measures. 
\end{example}

\noindent Our theory of extended Gaussians avoids unnormalized measures altogether: The nondeterministic distribution $X \sim \R$ is used as the uninformative prior on $X$, which gives the desired results, and $\R$ can be seen as the limit of $\mathcal N(50,\sigma^2)$ for $\sigma^2 \to \infty$ in an appropriate sense. 

\subsection{Contribution}
The paper is devoted to making our manipulations of subspaces as generalized probability distributions rigorous. We introduce a class of mathematical objects called \emph{extended Gaussian distributions} and show that such distributions can be manipulated (combined, pushed forward, marginalized) as if they were ordinary probability distributions. Importantly, extended Gaussians remain closed under taking conditional distributions, which means we can use them in applications such as statistical learning and Kalman filtering. The subspace $\R$, seen as a uniform distribution, formalizes the role of an improper prior. \\

\noindent Describing distributions on a space $X$ is only the first step. In order to build up systems in a compositional way, we need to understand transformations between spaces $X \to Y$. Category theory is a widely used language to study the composition of different kinds of systems. We identify two relevant flavors in the literature
\begin{itemize}
\item categorical and diagrammatic methods for engineering and control theory, such as graphical linear algebra (e.g. \cite{graphical_la}), cartesian bicategories (e.g. \cite{carboni:cartesian_bicategories}) and signal-flow diagrams (\cite{bsz,bonchi:cat_signal_flow,bonchi:interacting_hopf,baez:control,baez:props}). A central notion is that of a \emph{hypergraph category} \cite{hypergraphcats}, and prototypical models are the categories of linear maps or linear relations. Willems' system theory has been explored in these terms \cite{fong:open_systems}, but probability is absent from these developments. 
\item categorical models of probability, such as copy-delete categories \cite{cho_jacobs} and Markov categories \cite{fritz,fritz:definetti,fritz:kantorovich,fritz:zero_one}. Prototypical models are stochastic matrices, or the category $\gauss$ of affine-linear maps with Gaussian noise.
\end{itemize}

\noindent Despite these developments, it has been challenging to combine probability and nondeterminism into a single model -- mathematical obstructions to achieving this are described in \cite{nogo,weakdist}. Our work is a first successful step in combining these bodies of literature: We define a category $\gaussex$ of extended Gaussian maps which can seen both as extending linear relations with probability, or extending Gaussian probability with nondeterminism (or improper priors). Gaussian probability is a very expressive fragment of probability with a variety of useful applications (Gaussian processes, Kalman filters, Bayesian Linear Regression). 

Our definition of the Markov category $\gaussex$ uses a special case of the widely studied construction of \emph{decorated cospans} \cite{fong:decorated_cospans,fong:decorated_corelations,fong:sevensketches,fong:open_systems}. We recall that $\gaussex$ has conditionals, which is the categorical formulation of conditional probability used to formalize inference problems like \Cref{ex:inference}.

We then define a hypergraph category of \emph{Gaussian relations}, which allows arbitrary decorated cospans to allow the possibility of failure and explicit conditioning in the categorical structure. Hypergraph categories are highly symmetrical categories with an appealing duality theory. To our knowledge, probabilistic models of hypergraph categories are novel. The self-duality of hypergraph categories is reflected in the duality between covariance and precision forms, which takes a particularly canonical form for extended Gaussians. We elaborate this in \Cref{sec:variance_precision}.

The following table summarizes the relationships between our constructions

\begin{center}
\begin{tabular}{ |c|c|c| } 
\hline
& & (adding Gaussian noise) \\
 \hline
 & linear maps & Gaussian maps  \\ 
 (adding nondeterminism) & total linear relations & extended Gaussian maps  \\ 
 (adding failure) & linear relations & Gaussian Relations  \\ 
 \hline
\end{tabular}
\end{center}

\subsection{Outline}

We assume basic familiarity of the reader with linear algebra, (monoidal) categories and string diagrams; an overview can be found in the appendix \Cref{sec:appendix}. All categories considered will be symmetric monoidal and have a copy-delete structure. All vector spaces are assumed finite dimensional. 

We begin \Cref{sec:gaussex_distributions} with a recap of Gaussian probability and continue to define extended Gaussian distributions as Gaussian distributions on quotient spaces. We extend this definition to a notion of \emph{extended Gaussian map} in \Cref{sec:gaussex} and establish the structure of a Markov category. We give the construction both in elementary terms and using the formalism of decorated cospans in \Cref{sec:cospans}. 

In \Cref{sec:relations}, we define a hypergraph category of Gaussian relations, which extends extended Gaussian maps with the possibility of failure and conditioning. This makes use of the discussion of conditionals in \Cref{sec:conditioning}.

The idea of extended Gaussian distributions has appeared in several places independently, for different motivations. We conclude the paper with an extended `Related Works' \Cref {sec:connections}, which compares these approaches in detail, and gives perspective in terms of measure theory, topology, and program semantics. 

\section{Extended Gaussian Distributions}\label{sec:gaussex_distributions}

We begin with a short review of Gaussian probability; we assume basic concepts of linear algebra but have summarize the terminology in the appendix (\Cref{app:la}). For a more detailed introduction to Gaussian probability see e.g. \cite{alexthesis,lauritzen}. \\

\noindent The normal distribution or Gaussian distribution $\mathcal N(\mu,\sigma^2)$ of mean $\mu$ and variance $\sigma^2$ is defined by having the density function
\[ f(x) = \frac{1}{\sqrt{2\pi\sigma^2}} \exp\left(-\frac{(x-\mu)^2}{2\sigma^2}\right) \]
with respect to the Lebesgue measure. This is generalized to multivariate normal distributions as follows: Every \emph{Gaussian distribution} on $\R^n$ can be written uniquely as $\N(\mu, \Sigma)$ where $\mu \in \R^n$ is its \emph{mean} and $\Sigma \in \R^{n \times n}$ is a symmetric positive-semidefinite matrix called its \emph{covariance matrix}. Note that a vanishing covariance matrix is explicitly allowed; in that case the Gaussian reduces to a point-mass $\delta_x = \N(x,0)$. We will sometimes abbreviate the point-mass $\delta_x$ by $x$ if the context is clear. 

We write $\gauss(\R^n)$ for the set of all Gaussian distributions on $\R^n$. The \emph{support} of $\N(\mu,\Sigma)$ is the affine subspace $\mu + \col(\Sigma)$ where $\col(\Sigma)$ is the column space (image) of $\Sigma$. Gaussian distributions transform as follows under linear maps: If $A \in \R^{m \times n}$ is a matrix, then the pushforward distribution is given by
\begin{equation} A_*(\N(\mu, \Sigma)) = \N(A\mu, A\Sigma A^T) \label{eq:gauss_pushfwd} \end{equation}
Product distributions are formed as follows
\begin{equation} \N(\mu, \Sigma) \otimes \N(\mu', \Sigma') = \N\left(\begin{pmatrix} \mu \\ \mu' \end{pmatrix},
\begin{pmatrix} \Sigma & 0 \\ 0 & \Sigma \end{pmatrix}\right) \label{eq:gauss_product} \end{equation}
We write addition $+$ between distributions to indicate the distribution of the sum of two independent variables (convolution). For example, if $X,Y \sim \N(0,1)$ are independent, then $X + Y \sim \N(0,2)$ because variance is additive for independent variables. We have
\[ \N(\mu, \Sigma) + \N(\mu',\Sigma') = \N(\mu + \mu', \Sigma + \Sigma') \]
which can be confirmed by first forming the product distribution \eqref{eq:gauss_product} and pushing forward under the addition map \eqref{eq:gauss_pushfwd}. The set $\gauss(\R^n)$ forms a commutative monoid with convolution $+$ and neutral element $0$. \\

We now wish to a combine Gaussian distributions on $\R^n$ with uninformative (nondeterministic) distributions along a vector subspace $D$. 
\begin{definition}
An extended Gaussian distribution on $\R^n$ is a pair $(D,\psi)$ of a subspace $D \subseteq \R^n$ and a Gaussian distribution $\mu$ on the quotient $\R^n/D$. Following \cite{willems:oss}, we call the space $D$ the \emph{(nondeterministic) fibre} of the extended Gaussian. We write $\gaussex(\R^n)$ for the set of all extended Gaussian distributions on $\R^n$.
\end{definition}
There are several equivalent ways to formalize the notion of a Gaussian distribution over this quotient space.
\begin{enumerate}
\item We identify the quotient space $\R^n/D$ with a complementary subspace $K$ of $D$, and give a Gaussian distribution on that space. This has the advantage of only involving Euclidean spaces, and we can use matrices to represent linear maps. 
\item We develop a coordinate-free definition of Gaussian distributions on arbitrary vector spaces $X$ so we can then interpret the construction $\gauss(\R^n/D)$ directly. This will be useful for the duality results in \Cref{sec:variance_precision}.
\item Willems keeps the spaces $\R^n$ but equips them with restricted $\sigma$-algebras. This corresponds to a quotient on the level of measurable spaces. We discard this perspective for now but will return to it in \Cref{sec:open_system}.
\end{enumerate}

For now, it doesn't matter which formalization we choose. We will build intuitions with some examples: \vspace{-0.2cm}

\begin{enumerate}
\item Every Gaussian distribution $\psi$ becomes an extended Gaussian distribution with $D=0$; that is the nondeterministic contribution vanishes (is constantly zero).
\item Every subspace $D$ becomes an extended Gaussian distribution with $\psi=0$; that is the probabilistic contribution vanishes. By slight abuse of notation, we will simply write $D$ or $\psi$ for the embedding of subspaces or distributions into extended Gaussian distributions.
\item If the nondeterministic fibre $D=\R^n$ is the whole space, then $\gauss(\R^n/D) = \{0\}$. Hence, the only extended Gaussian with fibre $D$ is the subspace $D$ itself. This distribution expresses total ignorance. 
\item We can easily classify all extended Gaussian distributions on $\R$. The fibre $D$ must be either $0$ or $\R$, so we have $\gaussex(\R) = \gauss(\R) \,\dot{\cup}\, \{ \R \}$.
\item The possible pairs $(V,I)$ satisfying Ohm's law are given by the subspace $D=\{(V,I) : V=RI \}$. For noisy Ohm's law, we let $\epsilon \sim \mathcal N(0,\sigma^2)$ and notice that the random vector $w = (-1,R) \cdot \epsilon$ is orthogonal to $D$. Its covariance matrix is
\[ \Sigma_w = \begin{pmatrix}
    1 & -R \\ -R & R^2
\end{pmatrix}\]
and thus the distribution of the noisy law is given by $(\mathcal N(0,\Sigma_w),D)$. 
\end{enumerate}

\noindent We may think of the extended distribution $(D,\psi)$ as being composed of nondeterminstic noise along the space $D$, and Gaussian noise $\psi$. It is evocative to write the extended Gaussian distribution as a formal sum $\psi + D$ of distribution and a subspace. The distribution $\psi$ is not unique because the nondeterministic noise absorbs components of $\psi$ that are parallel to $D$. This is analogous to how we use notation like $3 + 2\mathbb Z$ for elements of quotient groups (cosets). This notation is formally justified by the formula for addition of extended Gaussians, as discussed next.

\subsection{Transformations of Extended Gaussians}

Extended Gaussian distributions support the same basic transformations as ordinary Gaussians. 
If $A$ is a matrix, we push forward the Gaussian and nondeterministic contribution separately,
\[ A_*(\psi + D) = A_*\psi + A[D] \]
where $A[D] = \{ Ax : x \in D \}$ denotes the image subspace. Tensor and sum are similarly component-wise
\begin{align*}
(\mu + D) \otimes (\psi + E) = (\mu \otimes \nu) + (D \times E) \qquad
(\mu + D) + (\nu + E) = (\mu + \nu) + (D + E) 
\end{align*} 
Well-definedness is a corollary of the next section, because those operations are special cases of the categorical structure of $\gaussex$.

\begin{example}\label{ex:r_inv}
The subspace $\R \in \gaussex(\R)$ absorbs all additive contributions, e.g. $42 + \R = \N(0,1) + \R = \R$
\end{example}

\section{A Category of Extended Gaussian maps}\label{sec:gaussex}

After defining extended Gaussians on Euclidean spaces $X$, the next challenge is to develop a notion of \emph{extended Gaussian map} $X \to Y$ between spaces. We wish to define a category $\gaussex$ such that we recover extended Gaussian distributions as maps out of the unit space $0$, i.e. $\gaussex(X) \cong \gaussex(0, X)$.
The operations of pushforward, product and sum of distributions will be simple instances of categorical and monoidal composition in the category $\gaussex$. For purely Gaussian probability, the appropriate definition of a map is a linear function together with Gaussian noise, informally written $f(x) = Ax + \mathcal N(b,\Sigma)$. We begin by analyzing this construction before generalizing it to the extended Gaussian case.

\subsection{Decorated Linear Maps and the Category $\gauss$}
We write $\vect$ for the category of finite dimensional vector spaces. The category $\gauss$ \cite{fritz} is defined as follows: Objects are vector spaces $X$, and morphisms $X \to Y$ are pairs $(f,\psi)$ of a linear map $f : X \to Y$ and a Gaussian distribution $\psi \in \gauss(Y)$. The identity is given by $(\id_X, 0)$ and composition is given by pushing forward and addition of the noise, $(f,\xi) \circ (g,\psi) = (fg,\xi + f_*\psi)$.

It is straightforward to generalize the pattern of this construction: The set of distributions $\gauss(X)$ is a commutative monoid $(\gauss(X),+,0)$ and the assignment $X \mapsto \gauss(X)$ becomes a lax monoidal functor $\gauss : \vect \to \cmon$ from vector spaces to commutative monoids. By understanding a commutative monoid as a one-object category, the functor $\gauss : \vect \to \mathsf{Cat}$ is an indexed category, and the category $\gauss$ is the monoidal op-Grothendieck construction associated to this functor \cite{moeller2018monoidal}.

We do not use any special properties of Gaussian distributions, other than that they can be added and pushed forward. In other words, can think of the distribution $\psi$ as a purely abstract decoration on the codomain of the linear map $f$. Any functor $S : \vect \to \cmon$ can be used to supply such a decoration, because it it automatically inherits a lax monoidal structure (see below). In concrete terms, the op-Grothendieck construction can be described as decorated linear maps:
\begin{definition}
Let $S : \vect \to \cmon$ be a functor. The category $\lin{S}$ of $S$-decorated linear maps is defined as follows
\begin{enumerate}
	\item Objects are vector spaces $X$
	\item Morphisms are pairs $(f,s)$ where $f : X \to Y$ is a linear map and $s \in S(Y)$
	\item Composition is defined as follows: for $g : X \to Y$, $f : Y \to Z$, $s \in S(Y)$, $t \in S(Z)$ let
	\[ (f,t) \circ (g,s) = (fg,t + S(f)(s)) \]
	Note that addition takes place in the commutative monoid $S(Z)$.
\end{enumerate}
\end{definition}

There is a faithful inclusion $\vect \to \lin S$ sending $f$ to $(f,0)$. We argue that $\lin S$ has the structure of a symmetric monoidal category with the tensor $X \otimes Y = X \times Y$ on objects. For this, we first observe that $S$ is automatically lax monoidal: For $(s,t) \in S(X) \times S(Y)$, let $s \oplus t = S(i_X)(s) + S(i_Y)(t)$ where $i_X : X \to X \times Y, i_Y : Y \to X \times Y$ are the biproduct inclusions. We can now define the tensor of decorated map as $(f,s) \otimes (g,t) = (f \times g, s \oplus t)$. The monoidal category $\lin S$ is in general not cartesian; it does however inherit copy and delete maps from $\vect$. The category $\lin S$ is a Markov category if and only if deleting is natural, i.e. $S(0) \cong 0$, where $0$ denotes the terminal vector space/commutative monoid.

\begin{example}\label{ex:bottomrow}
The following categories are instances of decorated linear maps:
	\begin{enumerate}
		\item For $S(X)=0$, $\lin S$ is equivalent to $\vect$.
		\item For $S(X)=X$, $\lin S$ is equivalent to the category of affine-linear maps. A map $X \to Y$ consists of a pair $(f,y)$ with $f : X \to Y$ linear and $y \in Y$.
		\item For $S(X)=\gauss(X)$, $\lin S$ is (by construction) the category $\gauss$
	\end{enumerate}
\end{example}

\subsection{Decorated Cospans and Linear Relations}\label{sec:cospans}

Like for $\gauss$, we wish to define an extended Gaussian map as a linear map with extended Gaussian noise. The naive approach of considering linear maps decorated by $S=\gaussex$ is not fruitful, because the quotient by the nondeterministic fibre is not properly taken into account: For example, for any two linear maps $f,g : \R^n \to \R$, the decorated maps $f + \R$ and $g + \R$ should be considered equal (\Cref{ex:r_inv}). We can remedy this by considering maps into the quotient $X \to \R/\R$. This kind of behavior is precisely captured by (total) linear relations. 

\begin{lemma}[{{\Cref{app:la}}}]  \label{lemma:linrel_characterization}
To give a total linear relation $R \subseteq X \times Y$ is to give a subspace $D \subseteq Y$ and a linear map $X \to Y/D$.
\end{lemma}

\begin{definition}
An extended Gaussian map $X \to Y$ is a tuple $(D,f,\psi)$ where $D \subseteq Y$ is a subspace, $f : X \to Y/D$ and $\psi \in \gauss(Y/D)$.
\end{definition}
In order to describe composition of such maps, it is convenient to use the formalism of decorated cospans, which we recall now:

A cospan in a category $\cat C$ with finite colimits is a diagram of the form $X \xrightarrow{f} P \xleftarrow{g} Y$. We will identify two cospans if there exists an isomorphism $P \cong P'$ commuting with the legs. Equivalence classes of cospans can be seen as morphisms between $X$ and $Y$ in a category $\mathsf{Cospan}(\cat C)$, where composition is given by pushout
\begin{equation}\begin{tikzcd}
	& X &&&&& W \\
	X && X &&& P && Q \\
	&&&& X && Y && Z
	\arrow[from=3-5, to=2-6]
	\arrow[from=3-7, to=2-6]
	\arrow[from=3-7, to=2-8]
	\arrow[from=3-9, to=2-8]
	\arrow[from=2-8, to=1-7]
	\arrow[from=2-6, to=1-7]
	\arrow["\lrcorner"{anchor=center, pos=0.125, rotate=-45}, draw=none, from=1-7, to=3-7]
	\arrow["{\id_X}", from=2-1, to=1-2]
	\arrow["{\id_X}"', from=2-3, to=1-2]
\end{tikzcd} \label{eq:cospan} \end{equation}

\noindent The following classes of cospans deserve special attention:
\begin{enumerate}
\item a cospan whose right leg is an isomorphism is the same thing as a map $X \to Y$
\item a relation is a span $X \leftarrow R \rightarrow Y$ which is jointly monic. Dually, a \emph{co-relation} is a cospan $X \rightarrow P \leftarrow Y$ which is jointly epic \cite{fong:decorated_corelations}.
\item a partial map is a span $X \leftarrow R \rightarrow Y$ whose left leg is monic \cite{cockett2002restriction}. Dually we define a \emph{copartial map} to be a cospan $X \rightarrow P \leftarrow Y$ whose right leg is epic.
\end{enumerate}
Just as partial maps are maps out of subobjects, copartial maps are maps into quotients. It is worth noting that while the pushout of copartial maps is again a copartial map, co-relations are not closed under pushout. Instead, the an image factorization has to be used to compose them \cite{fong:decorated_corelations}. \Cref{lemma:linrel_characterization} can be rephrased as follows:

\begin{proposition}
To give a copartial map $X \xrightarrow{f} P \xleftarrow{p} Y$ in $\vect$ is to give a total linear relation $X \to Y$. The relation is obtained as $R = \{ (x,y) : f(x) = p(y) \}$.
\end{proposition}

We now use the abstract theory of decorated cospans \cite{fong:decorated_cospans,fong:thesis,fong:sevensketches} to add Gaussian probability to the cospans:

\begin{definition}[{{\cite{fong:decorated_cospans}}}]
Given a lax monoidal functor $S : (\cat C,+) \to (\catname{Set},\times)$, an $S$-decorated cospan is a cospan $X \rightarrow P \leftarrow Y$ together with a decoration $s \in S(P)$. Given composable cospans like in \eqref{eq:cospan}, the decoration of the composite is computed by the canonical morphism $S(P) \times S(Q) \to S(P + Q) \to S(W)$. The category of $S$-decorated cospans is written $S\catname{Cospan}(\mathbb C)$
\end{definition}

The category $\gauss$ is a special case of the decorated cospan construction, for cospans whose right leg is an identity. We can now define:  

\begin{definition}
The category $\gaussex$ of extended Gaussian maps is defined as the category of copartial maps in $\vect$, decorated by the functor $\gauss : \vect \to \cmon \to \catname{Set}$.
\end{definition}

Categories of decorated cospans are hypergraph categories \cite[\S~2]{fong:decorated_cospans} their monoidal structure is given by the coproduct $+$. As the subcategory of decorated copartial maps, extended Gaussians do inherit the symmetric monoidal and copy-delete structure, but are not a hypergraph category. To obtain a useful hypergraph category of Gaussian probability, we must study conditioning.

\section{A Hypergraph Category of Gaussian Relations}\label{sec:relations}

A hypergraph category extends the structure of a copy-delete category in two important ways
\begin{enumerate}
    \item there is a multiplication $\mu_X : X \otimes X \to X$ on every object, which we think of as a comparison operation. It succeeds if both inputs are equal (and return the input), and fails otherwise. In linear relations, comparison is the relation $\{ (x,x,x) : x \in X \}$. Multiplication is dual to copying. In a probabilistic setting, we propose to think of the comparison as conditioning on equality. The `cap' $X \otimes X \to I$ is denoted as $\eq$ \cite{dilavore2023evidential, exact_conditioning}. 
    \item there is a unit $u_X : I \to X$ on every object, dual to deletion. The unit is neutral with respect to the multiplication, i.e. conditioning on the unit has no effect. This suggests we should think of the unit as a uniform distribution, or an improper prior. Both in linear relations and extended Gaussians, the unit is the subspace $X \subseteq X$.
\end{enumerate}

We arrive at the following synthetic dictionary for probabilistic inference and constraints in hypergraph categories:

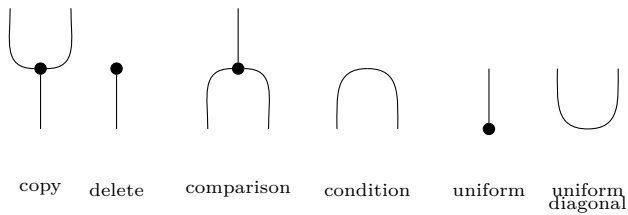
\begin{figure}[h]
    \centering
    \[ \begin{tikzpicture}[scale=0.4]
	\begin{pgfonlayer}{nodelayer}
		\node [style=none] (0) at (-8.5, -2) {};
		\node [style=none] (1) at (-6.5, -2) {};
		\node [style=bn] (2) at (-7.5, 0) {};
		\node [style=none] (4) at (-7.5, 2) {};
		\node [style=bn] (5) at (-14, 0) {};
		\node [style=none] (6) at (-13, 2) {};
		\node [style=none] (7) at (-15, 2) {};
		\node [style=none] (8) at (-14, -2) {};
		\node [style=none] (9) at (-11.5, -2) {};
		\node [style=bn] (10) at (-11.5, 0) {};
		\node [style=none] (13) at (-4.25, -2) {};
		\node [style=none] (14) at (-3.25, 0) {};
		\node [style=none] (15) at (-2.25, -2) {};
		\node [style=bn] (16) at (0.75, -2) {};
		\node [style=none] (17) at (0.75, 0) {};
		\node [style=none] (18) at (3, 0) {};
		\node [style=none] (19) at (4, -2) {};
		\node [style=none] (20) at (5, 0) {};
		\node [style=none] (23) at (-11.5, -4) {{\scriptsize delete}};
		\node [style=none] (24) at (-7.5, -4) {{\scriptsize comparison}};
		\node [style=none] (25) at (-3.25, -4) {{\scriptsize condition}};
		\node [style=none] (26) at (0.75, -4) {{\scriptsize uniform}};
		\node [style=none] (27) at (4, -4) {{\scriptsize uniform}};
		\node [style=none] (31) at (4, -4.5) {{\scriptsize diagonal}};
		\node [style=none] (33) at (-14, -4) {{\scriptsize copy}};
	\end{pgfonlayer}
	\begin{pgfonlayer}{edgelayer}
		\draw [in=90, out=-180, looseness=1.25] (2) to (0.center);
		\draw [in=90, out=0, looseness=1.25] (2) to (1.center);
		\draw (4.center) to (2);
		\draw [in=-90, out=180, looseness=1.25] (5) to (7.center);
		\draw [in=-90, out=0, looseness=1.25] (5) to (6.center);
		\draw (5) to (8.center);
		\draw (9.center) to (10);
		\draw [in=180, out=90, looseness=1.25] (13.center) to (14.center);
		\draw [in=90, out=0, looseness=1.25] (14.center) to (15.center);
		\draw (16) to (17.center);
		\draw [in=180, out=-90, looseness=1.25] (18.center) to (19.center);
		\draw [in=-90, out=0, looseness=1.25] (19.center) to (20.center);
	\end{pgfonlayer}
\end{tikzpicture} \]
    \caption{Dictionary for hypergraph categories}
    \label{fig:my_label}
\end{figure}

For example, the noisy measurement example \Cref{ex:inference} can be expressed in the following convenient way using hypergraph structure

    \[ \begin{tikzpicture}[scale=0.3]
	\begin{pgfonlayer}{nodelayer}
		\node [style=morphism] (0) at (-12.5, -2) {{\small $\mathcal N(50,100)$}};
		\node [style=bn] (1) at (-12.5, 0) {};
		\node [style=morphism] (2) at (-10.5, 2) {{\small $\mathcal N(-,25)$}};
		\node [style=none] (3) at (-14.5, 2) {};
		\node [style=none] (4) at (-10.5, 3) {};
		\node [style=none] (5) at (-8.5, 5) {};
		\node [style=none] (6) at (-6.5, 3) {};
		\node [style=none] (7) at (-14.5, 5) {};
		\node [style=morphism] (8) at (-6.5, -2) {{\small $40$}};
		\node [style=none] (9) at (-4.25, 1) {=};
		\node [style=morphism] (10) at (-1.75, -1.5) {{\small $\mathcal N(42,20)$}};
		\node [style=none] (11) at (-1.75, 3.5) {};
		\node [style=none] (12) at (-12.5, -1.5) {};
		\node [style=none] (13) at (-6.5, -1.5) {};
		\node [style=none] (14) at (-1.75, -1) {};
	\end{pgfonlayer}
	\begin{pgfonlayer}{edgelayer}
		\draw [in=-90, out=180] (1) to (3.center);
		\draw [in=-90, out=0] (1) to (2);
		\draw (7.center) to (3.center);
		\draw (2) to (4.center);
		\draw [in=-180, out=90] (4.center) to (5.center);
		\draw [in=90, out=0] (5.center) to (6.center);
		\draw (13.center) to (6.center);
		\draw (12.center) to (1);
		\draw (14.center) to (11.center);
	\end{pgfonlayer}
\end{tikzpicture}
 \]

We begin by recalling how conditioning works in the category $\gauss$, and prove that extended Gaussians remain closed under conditioning. We then define a hypergraph category $\catname{GaussRel}$ of Gaussian relations in which conditioning is internalized using a comparison operation.

\subsection{Conditioning}\label{sec:conditioning}

Gaussian distributions are self-conjugate; that is conditional distributions of Gaussians are themselves Gaussian. More precisely, given a joint distribution $\psi \in \gauss(X \times Y)$, the map which sends $x \mapsto \psi|_{X=x}$ is a Gaussian map $X \to Y$. This is captured using the following categorical definition: 

\begin{definition}[{{\cite[Definition~11.5]{fritz}}}]
A \emph{conditional} for a morphism $f : A \to X \otimes Y$ in a Markov category is a morphism $f_{|X} : X \otimes A \to Y$ which lets us reconstruct $f$ from its $X$-marginal as $f(x,y|a) = f_{|X}(y|x,a)f_X(y|a)$. In string diagrams, it satisfies
\[ \begin{tikzpicture}[scale=0.3]
	\begin{pgfonlayer}{nodelayer}
		\node [style=morphism] (0) at (-13.5, 1.5) {$f$};
		\node [style=none] (12) at (-14, 2) {};
		\node [style=none] (13) at (-13, 2) {};
		\node [style=none] (14) at (-14, 4.5) {};
		\node [style=none] (15) at (-13, 4.5) {};
		\node [style=none] (25) at (-10.5, 1.5) {=};
		\node [style=none] (26) at (-13.5, -1.5) {};
		\node [style=none] (27) at (-13.5, 1) {};
		\node [style=morphism] (28) at (-7.5, 0) {$f$};
		\node [style=none] (29) at (-8, 0.5) {};
		\node [style=none] (30) at (-7, 0.5) {};
		\node [style=bn] (31) at (-8, 3) {};
		\node [style=bn] (32) at (-7, 1.5) {};
		\node [style=bn] (33) at (-6.5, -2) {};
		\node [style=none] (34) at (-7.5, -0.25) {};
		\node [style=morphism] (35) at (-7, 5) {$f|_X$};
		\node [style=none] (36) at (-9, 4) {};
		\node [style=none] (37) at (-7, 7) {};
		\node [style=none] (38) at (-9, 7) {};
		\node [style=none] (39) at (-6.5, -3) {};
		\node [style=none] (40) at (-5.5, -0.5) {};
		\node [style=none] (41) at (-5.5, 2.5) {};
		\node [style=none] (42) at (-6.75, 4.5) {};
		\node [style=none] (43) at (-7.25, 4.5) {};
		\node [style=none] (44) at (-6.5, -4) {$A$};
		\node [style=none] (45) at (-13.5, -2.5) {$A$};
		\node [style=none] (46) at (-14, 5.5) {$X$};
		\node [style=none] (47) at (-13, 5.5) {$Y$};
		\node [style=none] (48) at (-9, 8) {$X$};
		\node [style=none] (49) at (-7, 8) {$Y$};
	\end{pgfonlayer}
	\begin{pgfonlayer}{edgelayer}
		\draw (13.center) to (15.center);
		\draw (12.center) to (14.center);
		\draw (26.center) to (27.center);
		\draw (30.center) to (32);
		\draw (29.center) to (31);
		\draw [in=-90, out=-180, looseness=1.25] (33) to (34.center);
		\draw [in=-90, out=-180, looseness=1.50] (31) to (36.center);
		\draw (38.center) to (36.center);
		\draw (37.center) to (35);
		\draw (39.center) to (33);
		\draw [in=-90, out=0] (33) to (40.center);
		\draw (40.center) to (41.center);
		\draw [in=-90, out=90] (41.center) to (42.center);
		\draw [in=0, out=-90] (43.center) to (31);
	\end{pgfonlayer}
\end{tikzpicture}
 \]
\end{definition}

\noindent The category $\gauss$ has all conditionals. By picking a convenient complement to the fibre $D$, we can reduce the problem of conditioning in $\gaussex$ to conditioning in $\gauss$.

\begin{theorem}\label{thm:cond}
$\gaussex$ has conditionals.
\end{theorem}
\begin{proof}
In the appendix (\Cref{app:cond}).
\end{proof}

\subsection{Gaussian Relations}\label{sec:condeq}

One difficulty of conditioning is that it introduces the possibility of failure. For example, the condition $0 \eq 1$ is infeasible. In general, given a joint distribution $\psi \in \gauss(X \times Y)$, we can only condition $\psi|_{X=x}$ if $x$ lies in the support of the marginal $\psi_X$. The dependence on supports is carefully analyzed in the `Cond' construction of \cite{2021compositional}.

We define a hypergraph category $\gaussrel$ of \emph{Gaussian relations} as follows
\[ \gaussrel(X,Y) \defeq \gaussex(X \times Y) + \{ \bot \} \]
That is, a Gaussian relation is either a joint extended Gaussian distribution, or a special failure symbol $\bot$ which represents infeasibility. Failure is strict in all categorical operations, i.e. composing or tensoring anything with failure is again failure. 

Most of the categorical structure of $\gaussrel$ is easy to define.
\begin{enumerate}
\item any morphism $f \in \gaussex(X,Y)$ can be embedded into $\gaussrel$ as its \emph{name} $\lceil f \rceil$ given by $I \xrightarrow{u_X} X \xrightarrow{\mathsf{copy}_X} X \otimes X \xrightarrow{\id_X \otimes f} X \otimes Y$  
\item the identity is the diagonal relation $D = \{ (x,x) : x \in X \} \in \gaussex(X \times X)$
\item copying and comparison are the both given by the relation $\{ (x,x,x) : x \in X \}$
\end{enumerate}

Composition of Gaussian relations requires conditioning: Given $R \in \gaussrel(X,Y)$ and $S \in \gaussrel(Y,Z)$, we compose them as follows: If any of them is $\bot$, return $\bot$. Otherwise form the tensor $R \otimes S \in \gaussex(X \otimes Y \otimes Y \otimes Z)$, and condition the two copies of $Y$ to be equal. If that condition is infeasible, return $\bot$. 

\subsection{Decorated cospans as generalized statistical models}\label{sec:cospan_model}

We can get a clearer view of composition $\gaussrel$ by using decorated cospans. Recall that decorated copartial functions $X \to P \leftarrow Y$ corresponded to extended Gaussian maps $X \to Y$. If we allow arbitrary cospans, we know that the category $\gauss\mathsf{Cospan}(\vect)$ is a hypergraph category by construction. We now explain how to view such a cospan as a kind of generalized statistical model, whose `solution' is a Gaussian relation.

\begin{theorem}\label{thm:span_cond}
We have a functor of hypergraph categories
\[ F : \gauss\mathsf{Cospan}(\vect) \to \gaussrel \]
which sends the decorated cospan $X \xrightarrow{f} P \xleftarrow{g} Y$ with decoration $\psi \in \gauss(P)$ to the Gaussian relation described by the solution to the following inference problem: Initialize $x \sim X$ and $y \sim Y$ with an uninformative prior. Then condition $f(x) - g(x) \eq \psi$, and return the posterior distribution in $\gaussex(X \times Y)$, or $\bot$ if the condition was infeasible.
\end{theorem}

Decorated cospans thus have an interpretation as a generalized kind of statistical model, and Gaussian relations can be understood as equivalence classes of such cospans which have the same solution. This approach is systematically explored with the Cond construction of \cite{2021compositional}, and indeed we can see $\gaussrel$ as a concrete representation of $\cond(\gaussex)$. 

\hide{
Let $\catname{C}$ be a Markov category with conditionals. It is tempting to think of conditioning in $\catname{C}$ as a higher-order function
$$
|_X: \catname{C}(A, X \otimes Y) \to \catname{C}(A \otimes X, Y)
$$
that maps morphisms to their conditionals. However, in many categories of interest, conditionals are non-unique. This is a severe difficulty for the category theorist interested in Bayesian inference. Tobias Fritz says,

\begin{quote}
Besides the previous two lemmas, we have not been able to identify any further potentially useful categorical properties of conditioning. Our original goal in this work had been to axiomatize conditioning as an operation on a category satisfying certain naturality conditions. We have abandoned this due to the lack of further candidate axioms for conditioning which would hold e.g. in FinStoch, and the well-known problem that conditionals are generally non-unique in cases where zero probabilities occur (Proposition 11.15), leading us to suspect that a reasonably coherent choice of conditionals is not possible.
\end{quote}
Conditioning, though central to Bayesian statistics, is generally a poorly-behaved operation. The category $\gaussex$ is no exception: though it has conditionals, they are not unique, and there is no coherent way to choose one conditional over another. We can get around this problem by thinking about "observation" instead of "conditioning".

\begin{definition} We say that extended Gaussian distributions $(\psi_1, D_1)$ and $(\psi_2, D_2)$ on $X$ have \textit{disjoint supports} if there does not exist any $x \in X$ such that
$$
q_1(x) \in \supp(\psi_1) \text{ and } q_2(x) \in \supp(\psi_2),
$$
where $q_1: X \to X/D_1$ and $q_2: X \to X/D_2$ are the canonical quotient maps.
\end{definition}
In particular, if $(\psi_2, D_2) = (\delta_x, \{0\})$ for some $x \in X$, then $(\psi_1, D_1)$ and $(\psi_2, D_2)$ have disjoint supports if and only if
$$
q_1(x) \notin \supp(\psi_1).
$$

\begin{definition}
For any vector spaces $X$ and $Y$, we define the function
$$
\mathrm{obs}: \gaussex(I, X \otimes Y) \times \gaussex(I, X) \to \gaussex(I, Y) \cup \{\bot\}
$$
for all $p: I \to X \otimes Y$ and $x: I \to X$ by
$$
\mathrm{obs}(p, x) = \bot
$$
if $x$ and $(\mathrm{id}_X \otimes \mathrm{del}_Y) \circ p$ have disjoint supports, and
$$
\mathrm{obs}(p, x) = p|_X \circ x
$$
otherwise.
\end{definition}
If we think of $p: I \to X \otimes Y$ as being a Bayesian model with parameter space $Y$ and sample space $X$, then the partial function
$$
\mathrm{obs}(p, -): \gaussex(I, X) \to \gaussex(I, Y)
$$
updates the model in response to data. If $p$ and $x$ do not have disjoint supports, then $\mathrm{obs}(p, x)$ is a posterior distribution over $Y$. Otherwise, no well-defined posterior distribution exists, and $\mathrm{obs}(p, x) = \bot$.

\subsection{Gaussian Relations and Decorated Cospans}

A Gaussian relation from $m$ to $n$, written $f: m \to n$, is an element of the set 
$$
\gaussex(I, \mathbb{R}^{m + n}) \cup \{\bot\}.
$$
The monoidal product of Gaussian relations $f_1: m_1 \to n_1$ and $f_2: m_2 \to n_2$ is the Gaussian relation
$$
f_1 \oplus f_2: m_1 + m_2 \to n_1 + n_2
$$
given by
$$
f_1 \oplus f_2 = \bot
$$
if $f_1 = \bot$ or $f_2 = \bot$, and
$$
f_1 \oplus f_2 = (\mathrm{id}_{m_1} \otimes \sigma_{n_1, m_2} \otimes \mathrm{id}_{n_2}) \circ (f_1 \otimes f_2)
$$
otherwise. The sequential composite of Gaussian relations $f: l \to m$ and $g: m \to n$ is the Gaussian relation
$$
g \circ f: l \to n
$$
given by
$$
g \circ f = \bot
$$
if $f = \bot$ or $g = \bot$, and
$$
g \circ f = \mathrm{obs}\left(A_*(f \otimes g), \delta_0\right)
$$
otherwise, where $A: \mathbb{R}^{l + 2m + n} \to \mathbb{R}^{l + m + n}$ is the linear transformation
$$
A(w, x, y, z) = (x - y, w, z).
$$
\begin{proposition}
The above data defines a PROP $\mathrm{GaussRel}$ of Gaussian relations.
\end{proposition}
Sequential composition in $\mathrm{GaussRel}$, as defined above, is a little abstract. To make it more concrete, we can represent morphisms in $\mathrm{GaussRel}$ as $F$-decorated cospans, where
$$
F: \mathrm{Mat} \hookrightarrow \mathrm{Gauss}
$$
is the PROP functor embedding $\mathrm{Mat}$, the category of real matrices, into $\mathrm{Gauss}$.

\begin{proposition}
    There exist PROP functors
    \begin{equation}
    	\begin{tikzcd}
            \mathrm{LinRel} && \mathrm{GaussRel} \\
            \\
            \mathrm{Cospan}(\mathrm{Mat}) && \mathrm{FCospan}(\mathrm{Mat})
            \arrow[from=3-1, to=1-1]
            \arrow[from=3-3, to=1-3]
            \arrow[from=1-1, to=1-3]
            \arrow[from=3-1, to=3-3]
    	\end{tikzcd}
    \end{equation}
\end{proposition}
The functor $\mathrm{FCospan(Mat)} \to \mathrm{GaussRel}$
maps decorated cospans 
$$
(L: m \to k, R: n \to k, \phi: I \to k)
$$
to Gaussian relations
$$
\mathrm{obs}((\psi, D), \delta_0),
$$
where
$$
\psi = \phi \otimes \delta_{0},
$$
and
$$
D = \{ (Ry - Lx, x, y) \mid (x, y) \in \mathbb{R}^{m + n} \}.
$$
The existence of such a functor proves that $\mathrm{GaussRel}$ is a hypergraph category.

\begin{proposition}
    $\text{GaussEx}$ is the causal subcategory of $\text{GaussRel}$.
\end{proposition}
}

\section{Related Work and Applications}\label{sec:connections}

\subsection{Open Linear Systems and $\sigma$-algebras}\label{sec:open_system}
Recall that a probability space is a tuple $(X,\mathcal E,P)$ of a set $X$, a $\sigma$-algebra $\mathcal E$ and a probability measure $P : \mathcal E \to [0,1]$. A random variable is a function $V : X \to \R$ which is $(\mathcal E, \mathcal B(\R))$-measurable, where $\mathcal B(\R)$ denotes the Borel $\sigma$-algebra.

Willems defines an \emph{$n$-dimensional linear stochastic system} to be a probability space of the form $(\R^n,\mathcal E,P)$ for which there exists a `fibre' subspace $D \subseteq \R^n$ such that the $\sigma$-algebra $\mathcal E$ is given by the Borel subsets of $\R^n/D$ in the following sense: Pick any complementary subspace $K$ with $K \oplus D = \R^n$. Then, the events $V \in \mathcal E$ are precisely Borel cylinders parallel to $D$, i.e. of the form $V = A + D$ for $A \in \mathcal B(K)$. As an aside, we might wonder in which sense the the algebra $\mathcal E$ is a quotient construction. The measurable projection $p : (\R^n,\mathcal E) \to (K,\mathcal B(K))$ is \emph{not} an isomorphism of measurable spaces; after all, the underlying function is not invertible. It is however an isomorphism in the category of probability kernels, namely the inclusion $i : K \to \R^n$ is an inverse when considered as a stochastic map. This is because the Dirac measures $\delta_{x}$ and $\delta_{ipx}$ are equal on $\mathcal E$. This phenomenon of `weak quotients' is nicely explained in \cite[Appendix~A]{moss_perrone_2023}. \\

A linear system is called \emph{Gaussian} if the measure $P$ on $K$ is a normal distribution. We notice that this agrees precisely with our definition of an extended Gaussian distribution on $\R^n$ with fibre $D$. A linear system is classical only if $D=0$, in the sense that only in this case the measure $P$ is defined on the whole algebra $\mathcal B(\R)$. In the case $D=\R^n$, the $\sigma$-algebra becomes $\mathcal E=\{\emptyset,\R^n\}$ and we cannot answer any nontrivial questions about the system (\Cref{ex:r_inv}). 

Willems gives explicit formulas for combining Gaussian linear systems (`tearing, zooming and linking') \cite{willems:oss}. These operations have been treated in categorical form in \cite{fong:open_systems} but not for probabilistic systems. One fundamental operation in Willems' calculus is the \emph{interconnection} of systems: Two probability systems $(X,\mathcal E_1,P_1), (X,\mathcal E_2,P_2)$ on the same state space $X$ are called \emph{complementary} if for all $E_1,E_1' \in \mathcal E_1$ and $E_2,E_2' \in \mathcal E_2$, we have 
\[ E_1 \cap E_2 = E_1' \cap E_2' \Rightarrow P_1(E_1)P_2(E_2) = P_1(E_1')P_2(E_2') \]
That is, the product of probabilities $P_1(E_1)P(E_2)$ depends only on the intersection $E_1 \cap E_2$. The two probability measures $P_1, P_2$ can now be joined together on the larger $\sigma$-algebra $\mathcal E = \sigma(\mathcal E_1 \cup \mathcal E_2)$ by defining $P(E_1 \cap E_2) = P(E_1)P(E_2)$. This is what happens in \Cref{ex:resistor} when connecting a noisy resistor to a voltage source: The underspecified $\sigma$-algebras gets enlarged and nondeterministic relationships become become probabilistic ones. It seems to us that interconnection is a special case of the composition of Gaussian relations. 

It is furthermore interesting that Willems uses the term \emph{open} for probability systems with an underspecified $\sigma$-algebra on $X$, while in category theory, we think of open systems as morphisms $Y \to X$. A remarkable feature is that the $\sigma$-algebra, which in measure-theoretic probability is considered a property of the objects in question (i.e. measurable spaces), is here part of the morphisms. The cospan perspective unifies this, for in a cospan $Y \xrightarrow{f} Q \xleftarrow{q} X$, we can equip $X$ with the $\sigma$-algebra generated by the quotient map $q$.

\subsection{Variance-Precision Duality}\label{sec:variance_precision}
We recall the coordinate-free description of Gaussian probability and use it to show that extended Gaussians are highly symmetric objects, which enjoy an improved duality theory over ordinary Gaussians (reflecting the hypergraph structure of Gaussian relations). This also points towards future research to understand $\gaussex$ as a topological completion of ordinary Gaussians. Work in this direction is the variance-information manifold of \cite{james:variance_manifold}. For simplicity, we will consider only Gaussians of mean zero.

If the covariance matrix $\Sigma \in \R^{n \times n}$ is invertible, then its inverse $\Omega = \Sigma^{-1}$ is known as \emph{precision} or \emph{information matrix}. Precision is dual to covariance, in the sense that while covariance is additive for convolution $+$, precision is additive for conditioning.
\[ \Sigma_{\psi_1 + \psi_2} = \Sigma_{\psi_1} + \Sigma_{\psi_2}, \qquad \Omega_{\psi_1 \cap \psi_2} = \Omega_{\psi_1} + \Omega_{\psi_2} \]
The latter equation is reminiscent of logdensities, which add when conditioning. Indeed, the precision matrix appears in the density function of the multivariate Gaussian distribution $f(x) \propto \exp\left(-\frac 1 2(x-\mu)^T\Omega(x-\mu)\right)$. 

If we allow singular covariance matrices $\Sigma$, we still have well-defined Gaussian distributions albeit with non-full support; however the information matrix ceases to exist (and the distribution no longer has a density with respect to the $n$-dimensional Lebesgue measure). Not only does this break the duality, but we are left to wonder which kind of distribution corresponds to singular precision matrices: The answer is extended Gaussian distributions with nonvanishing fibre.    

In a coordinate-free way, the covariance of a distribution $\psi \in \gauss(X)$ is the bilinear form on the dual space $\Sigma : X^* \times X^* \to \R$ given by $\Sigma(f,g) \defeq \mathbb E[f(U)g(U)]-\mathbb E[f(U)]\mathbb E[g(U)]$. This form is symmetric and positive semidefinite. The precision form $\Omega$ is instead of type $\Omega : X \times X \to \R$. The duality between the two forms can be stated as follows:

\begin{theorem}\label{exduality}
	The following data are equivalent for every f.d.-vector space $X$
	\begin{enumerate}
		\item pairs $\langle S,\Omega \rangle$ of a subspace $S \subseteq X$ and a bilinear form $\Omega : S \times S \to \R$
		\item pairs $\langle F, \Sigma \rangle$ of a subspace $F \subseteq X^*$ and a bilinear form $\Sigma : F \times F \to \R$
	\end{enumerate}
\end{theorem}
At the core of this duality lies the notion of the \emph{annihilator} of a subspace, here denoted $(-)^\bot$. In brief, the correspondences are as follows
\begin{center}
\begin{tabular}{ c|cc } 
precision & $S = \ker(\Sigma)^\bot$ & $\ker(\Omega)=D$ \\
\hline
covariance & $F=D^\bot$ & $\ker(\Sigma) = S^\bot$
\end{tabular}
\end{center}
We give a proof of the duality in the appendix (\Cref{app:la}).

\subsection{Statistical Learning and Probabilistic Programming}

It is unsurprising that notions equivalent to extended Gaussians have appeared in the statistics (e.g. in \cite{hedegaard:gaussian_random_fields}). A novel perspective on statistical inference which more closely matches the categorical semantics is \emph{probabilistic programming}, a powerful and flexible paradigm which has gained traction in recent years (e.g. \cite{van_de_meent:ppl_introduction,staton:commutative_semantics,probmods2}). In \cite{2021compositional}, we argued that the exact conditioning operation (conditioning on equality) described in \Cref{sec:condeq} is a fundamental primitive in such programs, and enjoys good logical properties. We presented a programming language for Gaussian probability featuring a first-class exact conditioning operator $(\eq)$, with Python/F\# implementations available under \cite{gaussianinfer}. For example, the noisy measurement example expressed as a probabilistic program reads

\lstset{style=python_ppl}
\begin{lstlisting}
x = normal(50, 100)
y = normal(x, 25)
y =:= 40
return x
\end{lstlisting}

This language uses Gaussian distributions only, but it can effortlessly be extended to use extended Gaussian distributions, which are likewise closed under conditioning (\Cref{thm:cond}).

The behavior of the conditioning operator $(\eq)$ can be quite subtle, and it is difficult to decide when two programs are observationally equivalent. The denotational semantics defined in \cite{2021compositional} on the basis of the category $\cond(\gauss)$ is fully abstract, but it is still lacking a concrete description of when two different programs fragments have the same behavior in all contexts. This is remedied by passing to the concrete description of $\gaussrel$. In terms of \Cref{sec:cospan_model}, a program denotes a decorated cospan over $\vect$, and contextual equivalence is precisely the equivalence relation \Cref{thm:span_cond}. 

The correspondence between probabilistic programs and categorical models of probability (with conditioning) is elaborated in detail in \cite{dario_thesis}.

\bibliographystyle{acm}
\bibliography{main}

\section{Appendix}\label{sec:appendix}

\subsection{Glossary: Category Theory}\label{app:categories}
We assume basic familiarity of the reader with monoidal category theory and string diagrams. All relevant categories in this article are symmetric monoidal. 

A copy-delete category \cite{cho_jacobs} (or gs-monoidal category) is a symmetric monoidal category $(\C,\otimes,I)$ where every object $X$ is coherently equipped with the structure of a commutative comonoid, which is used to model copying ($\mathsf{copy}_X : X \to X \otimes X$) and discarding ($\mathsf{del}_X : X \to I$) of information. In string diagrams, the comonoid axioms are rendered as

\[ \begin{tikzpicture}[scale=0.3]
	\begin{pgfonlayer}{nodelayer}
		\node [style=none] (4) at (-3.25, -2) {};
		\node [style=none] (5) at (2.25, 0) {};
		\node [style=bn] (8) at (-3.25, -1) {};
		\node [style=none] (11) at (4.25, 2) {};
		\node [style=bn] (14) at (-2.25, 1) {};
		\node [style=none] (17) at (-4.25, 2) {};
		\node [style=none] (18) at (3.25, 2) {};
		\node [style=none] (20) at (-4.25, 0) {};
		\node [style=none] (22) at (2.25, 0) {};
		\node [style=none] (23) at (-3.25, 2) {};
		\node [style=none] (24) at (3.25, -2) {};
		\node [style=none] (28) at (-2.25, 0) {};
		\node [style=none] (30) at (4.25, 0) {};
		\node [style=none] (31) at (0, 0) {$=$};
		\node [style=bn] (32) at (2.25, 1) {};
		\node [style=none] (33) at (-2.25, 0) {};
		\node [style=bn] (34) at (3.25, -1) {};
		\node [style=none] (35) at (1.25, 2) {};
		\node [style=none] (36) at (-1.25, 2) {};
		\node [style=none] (37) at (11.25, 1) {};
		\node [style=none] (38) at (11.25, 2) {};
		\node [style=none] (39) at (9.25, 1) {};
		\node [style=none] (40) at (15, 0.5) {$=$};
		\node [style=none] (41) at (18.25, 1) {};
		\node [style=none] (42) at (10.25, -1) {};
		\node [style=bn] (43) at (9.25, 1) {};
		\node [style=bn] (44) at (17.25, 0) {};
		\node [style=none] (45) at (16.25, 1) {};
		\node [style=bn] (46) at (10.25, 0) {};
		\node [style=none] (47) at (17.25, -1) {};
		\node [style=none] (48) at (9.25, 1) {};
		\node [style=none] (49) at (18.25, 1) {};
		\node [style=none] (50) at (16.25, 2) {};
		\node [style=none] (51) at (13.75, 2) {};
		\node [style=bn] (52) at (18.25, 1) {};
		\node [style=none] (53) at (12.5, 0.5) {$=$};
		\node [style=none] (54) at (13.75, -1) {};
		\node [style=bn] (55) at (23.25, -0.25) {};
		\node [style=none] (56) at (22.25, 0.75) {};
		\node [style=none] (57) at (24.25, 0.75) {};
		\node [style=bn] (58) at (27.75, -0.25) {};
		\node [style=none] (59) at (26.75, 0.75) {};
		\node [style=none] (60) at (28.75, 0.75) {};
		\node [style=none] (61) at (25.5, 0) {$=$};
		\node [style=none] (62) at (23.25, -1.75) {};
		\node [style=none] (63) at (27.75, -1.75) {};
		\node [style=none] (64) at (22.25, 1.75) {};
		\node [style=none] (65) at (24.25, 1.75) {};
		\node [style=none] (66) at (26.75, 1.75) {};
		\node [style=none] (67) at (28.75, 1.75) {};
	\end{pgfonlayer}
	\begin{pgfonlayer}{edgelayer}
		\draw [style=none] (4.center) to (8);
		\draw [style=none, bend left=45] (8) to (20.center);
		\draw [style=none, bend right=45] (8) to (28.center);
		\draw [style=none] (33.center) to (14);
		\draw [style=none, bend left=45] (14) to (23.center);
		\draw [style=none, bend right=45] (14) to (36.center);
		\draw [style=none] (20.center) to (17.center);
		\draw [style=none] (24.center) to (34);
		\draw [style=none, bend left=45] (34) to (5.center);
		\draw [style=none, bend right=45] (34) to (30.center);
		\draw [style=none] (22.center) to (32);
		\draw [style=none, bend left=45] (32) to (35.center);
		\draw [style=none, bend right=45] (32) to (18.center);
		\draw [style=none] (30.center) to (11.center);
		\draw [style=none] (42.center) to (46);
		\draw [style=none, bend left=45] (46) to (48.center);
		\draw [style=none, bend right=45] (46) to (37.center);
		\draw [style=none] (43) to (39.center);
		\draw [style=none] (37.center) to (38.center);
		\draw [style=none] (47.center) to (44);
		\draw [style=none, bend left=45] (44) to (45.center);
		\draw [style=none, bend right=45] (44) to (49.center);
		\draw [style=none] (45.center) to (50.center);
		\draw [style=none] (52) to (41.center);
		\draw [style=none] (54.center) to (51.center);
		\draw [style=none, bend left=45] (55) to (56.center);
		\draw [style=none, bend right=45] (55) to (57.center);
		\draw [style=none, bend left=45] (58) to (59.center);
		\draw [style=none, bend right=45] (58) to (60.center);
		\draw (63.center) to (58);
		\draw (62.center) to (55);
		\draw (67.center) to (60.center);
		\draw (66.center) to (59.center);
		\draw [in=90, out=-90, looseness=0.75] (64.center) to (57.center);
		\draw [in=-90, out=90, looseness=0.75] (56.center) to (65.center);
	\end{pgfonlayer}
\end{tikzpicture}
 \]

Neither deleting nor copying are assumed to be natural in a copy-delete category. A \emph{Markov category} is a copy-delete category where deleting is natural, or equivalently, the monoidal unit $I$ is terminal. Markov categories typically model probabilistic or nondeterministic computation without possibility of failure, such as stochastic matrices, $\gauss$ or total (linear) relations. 

Copy-delete categories can model unnormalized probabilistic computation, or the potential of failure. The categories of partial functions or (linear) relations are typical examples of copy-delete categories that are not Markov categories. 

A \emph{hypergraph category} \cite{hypergraphcats} is a symmetric monoidal category with a particularly powerful self-duality: Every object is equipped with a special commutative Frobenius algebra structure.

\subsection{Noisy measurement example}\label{sec:mmt}

\begin{example}
	We elaborate the noisy measurement example from the introduction. Formally, we introduce random variables
	\begin{align*}
		X &\sim \N(50, 100) \\
		Y &\sim \N(X, 25)
	\end{align*}
	The vector $(X,Y)$ is multivariate Gaussian with mean $(50,50)$ and covariance matrix
	\[ \Sigma = \begin{pmatrix}
		100 & 100 \\
		100 & 125
	\end{pmatrix} \]
	The conditional distribution $X|(Y=40)$ is $\N(\mu=42,\sigma^2 = 20)$.
\end{example}
\begin{proof}
	The random vector $(X,Y)$ has joint density function
	\[ f(x,y) = \frac{1}{2\pi \cdot \sqrt{100 \cdot 25}} \exp\left(-\frac{(x-50)^2}{2 \cdot 100}\right) \cdot \exp\left(-\frac{(y-x)^2}{2 \cdot 25}\right) \]
	The conditional density of $x$ given $y$ has the form
	\[ f(x|y) = \frac{f(x,y)}{\int f(x,y) \mathrm d x} \]
	By expanding and `completing the square', it is easy to check that
	\[ f(x, 40) \propto \exp\left(-\frac{(x-50)^2}{200} -\frac{(40-x)^2}{50}\right) \propto \exp\left(-\frac{(x-42)^2}{2\cdot 20}\right) \]
	is again a Gaussian density, from which we read off $\mu = 42$ and $\sigma^2 = 20$. 
\end{proof}

\subsection{Glossary: Linear Algebra}\label{app:la}

All vector spaces in this paper are assumed finite dimensional. For vector subspaces $U,V \subseteq X$, their \emph{Minkowski sum} is the subspace $U+V = \{ u + v : u \in U, v \in V \}$. If furthermore $U \cap V = 0$, we call their sum a \emph{direct sum} and write $U \oplus V$. A \emph{complement} of $U$ is a subspace $V$ such that $U\oplus V = X$. An \emph{affine subspace} $W \subseteq X$ is a subset of the form $x + U$ for some $x \in X$ and a (unique) vector subspace $U \subseteq X$. The space $W$ is called a \emph{coset} of $U$ and the cosets of $U$ organize into the quotient vector space $X/U = \{ x + U : x \in X \}$. \\

An affine-linear map $f : X \to Y$ between vector spaces is a map of the form $f(x) = g(x) + y$ for some linear function $g : X \to Y$ and $y \in Y$. Vector spaces and affine-linear maps form a category $\aff$. \\

A \emph{linear relation} $R \subseteq X \times Y$ is a relation which is also a vector subspace of $X \times Y$. We write $R(x) \defeq \{ y \in Y : (x,y) \in R \}$. A relation $R \subseteq X \times Y$ is called \emph{total} if $R(x) \neq \emptyset$ for all $x \in X$. 
Linear relations and total linear relations are closed under the usual composition of relations. We denote by $\linrel$ and $\linrel^+$ the categories whose objects are vector spaces, and morphisms are linear relations and total linear relations respectively. $\linrel$ is a hypergraph category, while $\linrel^+$ is a Markov category. \\

The following lemma is crucial for relating linear relations and cospans: Every left-total linear relation can be written as a `linear map with nondeterministic noise' $x \mapsto f(x) + D$.

\begin{proposition}\label{prop:linrel_characterization}
	Let $R \subseteq X \times Y$ be a left-total linear relation. Then
	\begin{enumerate}
		\item\label{it:r0} $R(0)$ is a vector subspace of $Y$
		\item\label{it:coset} $R(x)$ is a coset of $R(0)$ for every $x \in X$
		\item\label{it:lin} the assignment $x \mapsto R(x)$ is a well-defined linear map $X \to Y/R(0)$
		\item\label{it:bij} every linear map $X \to Y/D$ is of that form for a unique left-total linear relation
	\end{enumerate}
\end{proposition}
\begin{proof}
	For \ref{it:r0}, consider $y,y' \in R(0)$ (by assumption nonempty), then by linearity of $R$
	\[ (0,y) \in R, (0,y') \in R \Rightarrow (0,\alpha y + \beta y') \in R \]
	so $R(0)$ is a vector subspace. For \ref{it:coset}, we can find some $w \in R(x)$ and wish to show that $R(x) = w + R(0)$. Indeed if $y \in R(x)$ then $(x,y) - (x,w) = (0,y-w) \in R$ so $y-w \in R(0)$, hence $y \in w+R(0)$. Conversely for all $z \in R(0)$ we have $(x,w+z) = (x,w)+(0,z) \in R$ so $w+z \in R(x)$. This completes the proof that $R(x)$ is a coset. For \ref{it:lin}, the previous point shows that the map $\rho : x \mapsto R(x)$ is a well-defined map $X \to Y/R(0)$. It remains to show it is linear. That is, if $w \in R(x)$ and $z \in R(y)$ then $\alpha w + \beta z \in R(\alpha x + \beta y)$. This follows immediately from the linearity of $R$. For the last point \ref{it:bij}, given a linear map $f : X \to Y/V$ we construct the relation
	\[ (x,y) \in R \Leftrightarrow y \in f(x) \]
	which is left-total because $f(x) \neq \emptyset$. To see that $R$ is linear, let $(x,y) \in R, (x',y') \in R$ meaning $y - z \in V$ and $y' - z \in V$ for representatives $z,z'$ of $f(x),f(x')$. Linearity of $f$ means that $\alpha z + \beta z'$ is a representative of $f(\alpha x + \beta x')$. Thus
	\[ \alpha y + \beta y' - (\alpha z + \beta z') = \alpha(y-z) + \beta (y'-z') \in V \]
\end{proof}

\subsection{Annihilators}

For subspaces $D \subseteq X$ and $F \subseteq X^*$, the subspaces $D^\bot \subseteq X^*, F^\bot \subseteq X$ are defined as
\begin{align}
	D^\bot \defeq \{ f \in X^* : f|_D = 0 \}, \qquad
	F^\bot \defeq \{ x \in X : \forall f \in F, f(x) = 0 \}
\end{align}

\begin{proposition}\label{prop:annih}
	\begin{enumerate}
		\item Taking annihilators is order-reversing and involutive
		\item If $D \subseteq S \subseteq X$, then $S^\bot \subseteq D^\bot \subseteq X^*$ and we have a canonical isomorphism
		\begin{equation}
			(S/D)^* \cong D^\bot/S^\bot \label{iso1}
		\end{equation}
		and similarly for $K \subseteq F \subseteq X^*$, we have
		\begin{equation} (F/K)^* \cong K^\bot/F^\bot \label{iso2} \end{equation}
		\item We have
		\[ (V + W)^\bot = V^\bot \cap W^\bot \]
		\[ (F \cap W)^\bot = F^\bot + G^\bot \]
		If $D \subseteq X$ and $f : X \to Y$, then 
		\[ (f[D])^\bot = \{ g \in Y^* : gf \in D^\bot \} \]
		If $U \subseteq X, V \subseteq Y$, we have a canonical isomorphism
		\[ (U \times V)^\bot \cong U^\bot \times V^\bot \]
	\end{enumerate}		
\end{proposition}
\begin{proof}
	Standard. An explicit description of the canonical iso \eqref{iso1} is given as follows.
	\begin{enumerate}
		\item We define $\alpha : D^\bot/S^\bot \to (S/D)^*$ as follows. If $f \in D^\bot$, then $f$ is a function $X \to \R$ such that $f|_D = 0$. The restriction $f|_S : S \to \R$ thus descends to the quotient $S/D \to \R$, and we let $\widetilde \alpha(f) = f|_S$. To check this is well-defined, notice that the kernel of $\widetilde \alpha$ consists of those $f \in X^*$ such that $f|_S=0$, that is $S^\bot$.
		\item We define $\alpha^{-1} : (S/D)^* \to D^\bot/S^\bot$ as follows. An element $f \in (S/D)^*$ is a function $f : S \to \R$ with $S|_D = 0$. Find any extension of $f$ to a linear function $\bar f : X \to \R$ (such an extension exists because $S$ is a retract of $X$). Then still $\bar f|_D = 0$, so $\bar f \in D^\bot$. It remains to show that the choice of extension does not matter in the quotient $D^\bot/S^\bot$. Indeed if $\bar f_2$ is another extension, then $(\bar f - \bar f_2)|_S = f - f = 0$, hence $(\bar f - \bar f_2) \in S^\bot$.
	\end{enumerate}
\end{proof}

\subsection{Conditionals}\label{app:cond}

The existence proof of conditionals in $\gaussex$ relies on the ability to pick a convenient complement to a subspace, as constructed by the following lemma:
\begin{lemma}\label{lemma:nicecomplement}
	Let $V \subseteq X \times Y$ be a vector subspace, and let $V_X \subseteq X$ be its projection. Then there exists a complement $K \subseteq X \times Y$ of $V$ such that $K_X$ is a complement of $V_X$.
\end{lemma}
\begin{proof}
	We give an explicit construction, where in fact we can choose $K$ to be a cartesian product of subspaces $U \times W$. Let
	\[ V_X = \{ x : (x,y) \in V \} \qquad H = \{ y : (0,y) \in V \} \]
	We argue that if $U \oplus V_X = X$ and $W \oplus H = Y$, then $(U \times W) \oplus V = X \times Y$.	
	First we prove that $(U \times W) \cap V = 0$: Indeed, if $(u,w) \in V$ for $u \in U, w \in W$, then $u \in V_X$, but that implies $u=0$. So we know $(0,w) \in V$, i.e. $w \in H$. Thus $w=0$.
	
	It remains to show that we can write every $(x,y)$ as $(u+v_1,w+v_2)$ with $u \in U, w \in W$ and $(v_1,v_2) \in V$.
	\begin{enumerate}
		\item We can write $x = u+v_1$ with $u \in U$ and $v_1 \in V_X$.
		\item We claim that there exists a $b \in W$ such that $(v_1,b) \in V$. Because $v_1 \in V_X$, there exists some $b' \in Y$ such that $(v_1,b') \in V$. We now decompose $b' = b + h$ for $b \in W, h \in H$. By definition of $H$, we have $(0,h) \in V$, so $(v_1,b)=(v_1,b')-(0,h) \in V$. 
		\item Write $y = w' + h$ with $w' \in W, h \in H$ and define $w = w' - b$ and $v_2 = h + b$. Then we have $w \in W$ and $(v_1,v_2) = (v_1,b)+(0,h) \in V$, and as desired
		\[ (u,w) + (v_1,v_2) = (x, w' - b + h + b) = (x,w' + h) = (x,y). \qedhere \]
	\end{enumerate}
\end{proof}

We can now prove the existence of conditionals in $\gaussex$. 

\begin{proof}[Proof of \Cref{thm:cond}]
	Let $\varphi \in \gaussex(A, X \times Y)$ be given by $(D,\tilde f,\widetilde \psi)$ where $\tilde f : A \to (X \times Y)/D$, $D \subseteq X \times Y$ and $\psi \in \gauss((X \times Y)/D)$. By Lemma~\ref{lemma:nicecomplement}, we can pick a complement $K \subseteq X \times Y$ of $D$ such that $K_X$ is a complement of $D_X$ in $X$. Under the identification $(X \times Y)/D \cong K$, we replace $\tilde f,\widetilde \psi$ with $f : A \to X \times Y$ and $\psi \in \gauss(X \times Y)$. \\
	
	\noindent Now we consider the morphism $x \mapsto f(x) + \psi$ in $\gauss(A, X \times Y)$ and find a conditional $f|_X \in \gauss(X \times A, Y)$. Informally, this means we can obtain $(X_1,Y_1) \sim f(a) + \psi$ as follows:
	\begin{align*}
	X_1 &\sim f_X(a) + \psi_X, & Y_1 &\sim g(x,a) + \xi 
	\end{align*}
	Similarly we can use conditionals in $\linrel^+$ to find a linear function $h : X \to Y$ and a subspace $H \subseteq Y$ such that $(X_2,Y_2) \sim D$ can be obtained as
	\begin{align*}
	X_2 &\sim D_X, & Y_2 &\sim h(X_2) + H
	\end{align*}
	Thus a joint sample $(X,Y) \sim f(x) + \psi + D$ can be obtained using the following process
	\begin{align*}
	X_1 &\sim f_X(a) + \psi_X, & X_2 &\sim D_X, & X &= X_1 + X_2 \\
	Y_1 &\sim g(X_1,a) + \xi, & Y_2 &\sim h(X_2) + H, & Y &= Y_1 + Y_2
	\end{align*}
	By construction we have $X_1 \in K_X, X_2 \in D_X$. Because $K$ was chosen such that $K_X \oplus D_X = X$, we can extract the individual values of $X_1,X_2$ from $X$ via the projections $P_{K_X}, P_{D_X} : X \to X$. A conditional for $\varphi$ is thus given by the formula
	\[ \varphi|_X(x,a) = g(P_{K_X}(x),a) + h(P_{D_X}(x)) + \xi + H \]
\end{proof}

\end{document}